\documentclass[11pt]{amsart}
\usepackage[left=1in, right=1in, top=1in]{geometry}
\usepackage[toc,page]{appendix}
\usepackage{amsmath}
\usepackage{amsfonts}
\usepackage{amssymb}
\usepackage{amsthm}
\usepackage{xcolor}
\usepackage{mathtools} 
\usepackage[normalem]{ulem} 
\usepackage{comment} 
\newtheorem{theorem}{Theorem}

\newtheorem{lemma}[theorem]{Lemma}
\newtheorem{note}[theorem]{Note}
\newtheorem{conjecture}[theorem]{Conjecture}
\newtheorem{definition}[theorem]{Definition}

\newcommand{\C}{{\mathbb C}}
\newcommand{\CC}{\mathcal{C}}
\newcommand{\PP}{\mathcal{P}}
\newcommand{\mc}[1]{\mathcal{#1}}
\newcommand{\on}[1]{\operatorname{#1}}
\newcommand{\spann}{\on{span}}
\newcommand{\vrho}{V_r^{\rho}} 
\newcommand{\Ber}{\on{Ber}_X} 
\newcommand{\minpi}[1]{\iota(#1)}

\newcommand{\minpj}{\minpi{P_j}}

\newcommand{\minipi}{\minpi{P_i}}
\newcommand{\maxv}{\on{max}(v)}
\newcommand{\redw}{w^c}
\newcommand{\m}{w_{-j, +i'}}

\begin{document}

\title{The Bad Locus in the moduli of Super Riemann Surfaces with Ramond punctures}
\author{Ron Donagi}
\address{Department of Mathematics, University of Pennsylvania, Philadelphia, PA 19104}
\email{donagi@math.upenn.edu}
\author{Nadia Ott}
\address{Department of Mathematics, University of Pennsylvania, Philadelphia, PA 19104}
\email{ottnadia@sas.upenn.edu}

\date{January 12, 2023}

\begin{abstract}
The bad locus in the moduli of super Riemann surfaces with Ramond punctures parametrizes those super Riemann surfaces that have more than the expected number of independent closed holomorphic 1-forms. There is a super period map that depends on certain discrete choices. For each such choice, the period map blows up along a divisor that contains the bad locus. Our main result is that away from the bad locus, at least one of these period maps remains finite. In other words, we identify the bad locus as the intersection of the blowup divisors. The proof abstracts the situation into a question in linear algebra, which we then solve. We also give some bounds on the dimension of the bad locus.

\end{abstract}
\maketitle
\newpage
\tableofcontents
\thispagestyle{empty}

\section{Introduction}

The super vector space  $\Omega$ of closed holomorphic differentials  on a general super Riemann surface of genus $g$ with $2r$ Ramond punctures has dimension $g|r$. The super period map  as defined by Witten in \cite{WittenRamond}  sends this space to a $g|r$ dimensional  maximal isotropic subspace $P(\Omega)$ of period space,  which is a $2g|2r$-dimensional super vector space with a supersymplectic pairing. Unfortunately  there is a bad locus $B$ in the  moduli space parametrizing super Riemann surfaces \footnote{  The  supermoduli space of super Riemann surfaces with Ramond punctures is a Deligne-Mumford superstack   (\cite{CodogniModuli}, \cite{OVRamond}, \cite{BruzzoHernRamond}, \cite{Moose}) and so locally has an \'etale cover by superschemes, and a compatible system of universal curves. All our statements pertain to the stack, and can be easily translated into statements over the \'etale covers.  }  where the odd  dimension of $\Omega$ is strictly greater than $r$, so the period map fails to be injective there. This bad locus $B$ was studied by Witten in \cite{WittenRamond}, mostly in the case $r=1$. The goal of this note is to study the bad locus for higher $r$. 

When $r=1$, the bad locus has codimension 1 and it is the locus where Witten's super period map blows up. We will see in Theorem \ref{dim} that in general  the bad locus (or at least one of its components) has codimension $r$. The period map depends on some discrete choices: a symplectic basis of the integral homology, as in the bosonic case, plus an orientation  $\psi$ of each of the Ramond punctures.
Each  choice of $\psi$ determines a period matrix that blows up along a certain divisor $Y''_{\sigma}$.
Here  $\sigma = \sigma(\psi)$ is a sign choice, determined by $\psi$ but retaining less information than $\psi$: there are $2^{2r}$ choices of $\psi$, but only $2^r$ choices for  $\sigma$. The blowup divisor $Y''_{\sigma}$ depends only on ${\sigma}$.
Our main result, Theorem \ref{Main}, is that the intersection $Y$ of these divisors equals the bad locus $B$. The $Y''_{\sigma}$ therefore give a more-or-less explicit set of equations for the bad locus. (This story has a variant that depends additionally on the choice of a pairing $\nu$ of the $2r$ punctures. 
It  follows from our result that this intersection of the  $Y''_{\sigma}$ is independent of the choice of $\nu$.)

We review the setup and state the main result in section \ref{periods}. Our strategy for the proof is to abstract the situation and translate it to a question in linear algebra, Conjecture \ref{conj}. A special case of this conjecture, sufficient for our needs here, is proved in Appendix \ref{pf}. The full conjecture was proved in \cite{Nathan}.
The  reduction of the main theorem to this linear algebra result  is explained in section \ref{translation}. 

During the preparation of this work, Ron Donagi was supported in part by NSF grant DMS 2001673 and by Simons HMS Collaboration grant \#390287. Nadia Ott was supported in part by Simons HMS Collaboration grant \#390287.

\section{Periods of super Riemann surfaces with Ramond punctures} \label{periods}

A genus g super Riemann surface  $X$  with $2r$ Ramond punctures determines a curve $C$ of the same genus $g$ with a divisor $D=\Sigma_{i=1}^{2r} x_i$  consisting of $2r$ distinct labeled points and a twisted spin structure, i.e. a line bundle $L$ satisfying 
\begin{equation}
L^{\otimes 2} = K_C (D).	\label{spin}
\end{equation}
Riemann-Roch gives
\begin{equation} \label{Euler}
\chi(L)= h^0(L) - h^0(L(-D)) =r.
\end{equation}
For $r>0$ an easy general position argument shows that generically $h^0(L)=r$. 
This means that the set of super Riemann surfaces with $h^0(L)=r$ can be identified with  an open subset of the moduli space
\begin{equation} \label{M}
M := \{(C,D,L) \  | \  C \ \text{a curve of genus} \ g,  D \in C^{2r}, \    
 L \in {\text{Pic}}^{g+r-1}(C) \ {\text{satisfies \eqref{spin}}} \}.  
\end{equation}
Following Witten \cite{WittenRamond}, we refer to the complement 
\[
B := \{(C,D,L) | h^0(L) \geq r+1 \} \subset M.
\]
 of this open subset as the ``bad locus".

The super vector space $\Omega$ of closed 1-forms on $X$   has dimension $g|h^0(L)$. 
(Our sign convention is that when $X$ is split, the forms pulled back from the reduced curve $C$ are considered even, e.g, for  local coordinates $(z, \theta)$ on $X$,   $dz$ is even and $d \theta$ is odd. Both this and the opposite convention occur in the literature.)
In the split case, the even closed forms are pulled back from $H^0(C,K_C)$ while the odd closed forms are exact, of the form $da(z)\theta$ where $a=a(z)$ is a section of the twisted spin bundle $L$, cf. \cite{WittenRamond}, section 5.1, or our appendix \ref{oddwitten}.
So generically $\Omega$ is of dimension $g|r$.

We recall Witten's  superperiod map
\[
P: \Omega \to \Lambda= \Lambda_0 \oplus \Lambda_- \cong \C^{2g|2r}.
\]
This depends on some discrete choices.
The $2g$ even periods, just as in the classical theory, depend on the choice of a symplectic basis $(a_1, \dots, a_g; b_1,\dots,b_g)$ of $H_1(X, \mathbb{Z})$.
Since the residue gives a natural trivialization of the fiber of $K_C(D)$ at each puncture $x_i \in D$, 
equation  \eqref{spin} specifies a trivialization of the square $(L|{x_i})^2$ of the fiber of $L$ at  $x_i$. By an {\em orientation} $\psi_i$ of the puncture we mean a square root of this, i.e. a trivialization of $L|{x_i}$ that squares to the above trivialization of $K_C(D)$. The $2r$ odd periods depend on a choice of orientation at each of the $2r$ punctures.

These discrete choices determine identifications
\[
\Lambda_0 = H^1(C, \C)
\]
and
\[
\Pi \Lambda_1 = H^0(L/L(-D)) = \bigoplus_{i=1}^{2r} L_{x_i}. 
\]
The even part of the superperiod map is then restriction from $X$ to $C$ followed by integration on the symplectic basis $(a_1, \dots, a_g; b_1,\dots,b_g)$. The odd part  is just the restriction from $C$ to $D$:
\[
H^0(C,L) \to H^0(D,L) = H^0(L/L(-D)),
\]
cf. Appendix \ref{oddwitten}.

One bad thing about the ``bad" locus $B$ is that this superperiod map $P$ is not injective  there: On $X \in B$ there are  (odd)  closed $1$-forms with all periods equal to zero, as can be seen from the long exact sequence 
\begin{equation} 0 \to H^0(L(-D)) \to H^0(L) \overset{P_1}{\to} H^0(L/L(-D))  \to H^1(L(-D)) \to \dots 
\end{equation}  
where $h^0(C,L(-D)) \ge 1$.

The target of the superperiod map is period superspace $\Lambda \cong \C^{2g|2r}$. This comes equipped with a  super symplectic pairing  $Q$  (symplectic on $\C^{2g}$, orthogonal on $\C^{0|2r}$) which we may think of as a super version of the usual intersection pairing on $H^1(X, \mathbb{R})$. In the obvious notation,
\[
Q= \Sigma_{k=1}^g a_k \wedge b_k - \Sigma_{i=1}^{2r} (\psi_i)^2,
\]
where we recall that the orientations $\psi_i$ form a basis for $\Lambda_1$. 
The period map can be described by a $(g|r) \times (2g|2r)$ matrix $\PP$. In order to obtain a $(g|r) \times (g|r)$ period matrix, we need to split period space as a sum of two complementary Q-Lagrangian subspaces, $\Lambda = \Lambda_A \oplus \Lambda_B$. For the even part we use the span of $(a_1, \dots, a_g)$ for $\Lambda_{A,0}$
and of $(b_1,\dots,b_g)$ for $\Lambda_{B,0}$. 
For the odd part we fix a pairing $\nu$ of the $2r$ punctures, e.g. we could pair $x_{2i-1}$ with  $x_{2i}$. We then take 
$\Lambda_{A,1}$ as the span of $x_{2i-1}+\sqrt{-1} x_{2i}$ and 
$\Lambda_{B,1}$ as the span of $x_{2i-1}-\sqrt{-1} x_{2i}$. The period map then decomposes: $P=(P_A,P_B)$, with corresponding decomposition of the matrix, $\PP=(\PP_A,\PP_B)$.  The $(g|r) \times (g|r)$ period matrix is then $\PP_A^{-1} \PP_B$ . This is defined only where 
$\PP_A$ is invertible. It depends on: 
\begin{itemize}
\item the data $X=(C,D,L)$ underlying the super Riemann surface with  labeled Ramond punctures, 
\item  the symplectic basis $(a_1, \dots, a_g; b_1,\dots,b_g)$ of $H_1(X, \mathbb{Z})$, and 
\item the orientations $\psi_i$. 
\end{itemize}

Sometimes we may want to replace the specific pairing ($x_{2i-1}$ with  $x_{2i}$) by an arbitrary pairing $\nu$ and to emphasize the dependence on this choice. 
A pairing $\nu$ is specified by relabeling the $2r$ points as:
\[
x_{ij}, \ \ \  i=1,\dots ,r, \ \  \  j \in \pm := \{+,- \}.
\]
By gluing the points $(x_{i+}, x_{i-})$ in each pair we obtain a nodal curve $\overline{C}$ (of arithmetic genus $g+r$, with $r$ nodes).
Specifying the pairing of the points is equivalent to specifying the normalization map $\nu: C \to \overline{C}$.

Having specified the pairing $\nu$, 
a {\em sign choice} $\sigma$ is a specification of a spin bundle, aka theta characteristic, 
on the  nodal curve $\overline{C}$, \emph{i.e}, 
a line bundle $L_{\sigma}$ on the nodal curve $\overline{C}$ satisfying 
\begin{equation}\label{spinsing}
L_{\sigma}^{\otimes 2} = K_{\overline{C}}.	
\end{equation}
In particular, 
$L:=\nu^*L_{\sigma}$ is a twisted spin structure on $C$ satisfying \eqref{spin}. 
Equivalently, given $L$, $\sigma$ specifies for each $i=1,\dots,r$  an isomorphism 
\begin{equation} \label{sigma}
\sigma_i: L_{x_{i+}} \to L_{x_{i-}}
\end{equation}
whose square:
\[
\sigma_i^{\otimes 2}: K_C (D)_{|{x_{i+}}} = L_{{x_{i+}}}^{\otimes 2} \to  L_{{x_{i-}}}^{\otimes 2} = K_C (D)_{|{x_{i-}}}
\]
is identified via residues with multiplication by $-1$,
\[
\sigma_i^{\otimes 2} = -1: \C \to \C.
\]
($L_{\sigma}$  clearly determines a $\sigma$ as in \eqref{sigma}. Conversely, given $\sigma$ as in \eqref{sigma}, we recover $L_{\sigma}$ as the subsheaf of $L$ consisting of sections $s$ of $L$ satisfying 
\begin{equation} \label{gluings} s(x_{i-})= - \sqrt{-1} \sigma_i s(x_{i+}) \end{equation}  for all $1 \le i \le r$. 
In particular,  $h^0(\overline{C}, L_{\sigma}) >0$ if $L$ has a \emph{global} section satisfying the gluing equations \eqref{gluings} for all $1 \le i \le r$. ) 

Recall that an orientation $\psi$ is a trivialization of $L \vert_D$ compatible via residues with \eqref{spin}. In our current notation this amounts to trivializations $\psi_{ij}$ of each $L_{x_{ij}}$. Such a $\psi$ combined with formula \eqref{sigma} therefore determines a 
$\sigma=\sigma(\psi)$. Changing $\psi$ at both $x_{i+}$ and $x_{i-}$ leaves $\sigma$ unchanged. On a given $X$ there are $2^{2r}$ $\psi$'s but only $2^r$ $\sigma$'s. 

Recall that the odd parts
\[
P_1:\Omega_1 \to\Lambda_1, \ P_{A,1}:\Omega_1 \to \Lambda_{A,1}
\]
of the superperiod map depend on the orientation $\psi$. The kernels, however, depend only on $\sigma=\sigma(\psi)$. Indeed, if we identify $\Omega_1$ with $H^0(C,L)$, we find the identifications :
\begin{equation} \label{ker}
Ker(P_{1})  = H^0(L(-D)) \subset 
Ker(P_{A,1})  = H^0(\overline{C}, L_{\sigma}) \subset H^0(C,L).
\end{equation}



 We will write $\PP^{\nu, \sigma}$ when we want to emphasize  that the matrix $\PP$ representing the period map depends on a choice of pairing $\nu$ and sign $\sigma$.

\section{Moduli spaces and the main theorem} 

We are interested in describing the locus in moduli space where the period matrix is not defined, i.e., the locus where  $\PP_A$ is singular for all possible  signs $\sigma$ and for some or all pairings $\nu$.  To do this, we need to consider some covers of the moduli space $M$ of \eqref{M}.
Let 
\[
M' := \{\nu: C \to  \overline{C}, D, L \}
\]
be the (bosonic) moduli space parametrizing the data of a curve $C$ with  labeled divisor $D$, a pairing $\nu$ on the points of $D$, and a twisted spin bundle $L$ satisfying \eqref{spin}. There is a forgetful map 
\[
\pi: M' \to M. 
\]
There is a further $2^r$-sheeted cover 
\[
\pi': M'' \to M'
\] 
where 
\[
M'' := \{\nu: C \to  \overline{C}, D, L_{\sigma} \}
\]
parametrizes curves $C$ with labeled divisor $D$, a pairing of the points of $D$, and a  choice of a spin bundle $L_{\sigma}$ satisfying \eqref{spinsing}. 

Let  $\PP$ denote (the component of degree $g+r-1$ of) the relative Picard of the universal
curve over the Deligne-Mumford compactification of the moduli space $M_{g+r}$ of curves of genus $g+r$, so
\[
\PP := \{ C', L'  \ | \ C' \ {\text{a stable curve of genus}} \ g +r, L' \in {\text{Pic}}^{g+r-1}(C') \}.
\]
In $\PP$ there is a natural divisor $\PP_e$  parametrizing pairs $ \{ C', L' \}$ where $L'$ is an effective line bundle: $h^0(L')>0$.
There is a natural map of $M''$ to  $\PP$ sending
$\{\nu: C \to  \overline{C}, D, L_{\sigma} \}$
to
$\{  \overline{C}, L_{\sigma} \}.$
We let $Y''$ denote the inverse image of the divisor $\PP_e$:
\[
Y'' = Y'' _{\sigma} := \{  ( \nu: C \to  \overline{C}, D, L_{\sigma}) \  | \ L_{\sigma}^{\otimes 2} = K_{ \overline{C}}, \ h^0(L_{\sigma}) \geq 1 \} \subset M''.
\]
This is a divisor in $M''$.  If $(C,D,L, \nu, \sigma) \in Y''$ then  there exists a global section of $L$ satisfying \eqref{gluings} for all $1 \le i \le r$, i.e., there exists a global section $s$ whose odd A-periods all vanish. 
So $P_{A,1}^{\nu, \sigma}$ has a non-trivial kernel and the matrix $\PP_{A,1}^{\nu, \sigma}$ is singular.

As shown in Appendix A of  \cite{WittenRamond}, the period matrix (viewed as a map on $M''$) blows up (has a pole) along this divisor. 
This divisor $Y'' _{\sigma} $ does not depend on a choice of sign $\sigma$; the subscript $\sigma$ simply reminds us that it is a divisor in $M''$, whose fiber over a point of $M'$ is indexed by the $\sigma$'s.) The reason that $\sigma$ occurs here rather than the full orientation $\psi$ was explained at the end of the previous section:  $Y''$ parametrizes points where the (odd part of the) superperiod map fails to be injective. According to \eqref{ker}, this depends only on $\sigma$ rather than $\psi$.

Descending to $M'$, we have the closed locus
\[
Y' := \{  ( \nu: C \to  \overline{C}, D, L)  \  | \  L^{\otimes 2} = K_C (D),   \ h^0(L_{\sigma}) \geq 1 \ \forall \sigma \} 
= \{ y \in M'  \ | \  \pi'^{-1}(y) \subset Y''  \} 
\subset M'
\]
parametrizing super Riemann surfaces with a given pairing of their $2r$ Ramond punctures $D$ such that the period matrix blows up   for {\em every} sign choice $\sigma$, i.e., $(C,D,L, \nu) \in Y'$ if $ \PP_A^{\nu, \sigma}$ is singular for all $\sigma$. So the complement of $Y'$ is the open subset where the period matrix, for at least one sign choice $\sigma$, makes sense.
Somewhat informally we write $Y'=\cap_{\sigma} Y''_{\sigma}$. (Informal because $Y'' \subset M''$ while $Y' \subset M'$. )

Descending further to $M$, we have the closed locus
\[
Y := \{(C,D,L) | ( \nu: C \to  \overline{C}, D, L) \in Y' \ \forall \nu \} = \{ y \in M  \ | \ \pi^{-1}(y) \subset Y' \} \subset M,
\]
where $(C,D, L) \in Y$ if $\PP_A^{\nu, \sigma}$ is singular for all pairings $\nu$ and signs $\sigma$, 
and also the a priori larger locus
\[
\tilde{Y} := \{(C,D,L) | ( \nu: C \to  \overline{C}, D, L) \in Y' \ {\text{for some}} \ \nu \} =  \pi(Y') \subset M. 
\]
The distinction is that $(C,D,L)$ is in $\tilde{Y}$ if   $(C,D,L, \nu) \in Y'$ for some pairing $\nu$, while in order to be in $Y$ this must hold for all $\nu$ . 

\begin{lemma} \label{something} $B \subseteq Y \subseteq  \tilde{Y}$ \end{lemma}

\begin{proof} If $(C,D,L) \in B$, then $h^0(L(-D)) \ge 1$ and so there exists a global section of $L$ satisfying the gluing equation \eqref{gluings} for all $1 \le i \le r$ and for all pairings $\nu$ and signs $\sigma$. 
 \end{proof}

Our main result can now be stated as follows:

\begin{theorem} \label{Main}
$B = Y=\tilde{Y} $.
\end{theorem}

\section{Linear algebra}

Consider a set of $2r$ abstract points 
\[
p_{ij}, i=1,\dots ,r,  \  j \in \pm := \{+1,-1 \}.
\] 
By a section, or $r$-section, we mean a function 
\[
\sigma:\{1, \dots ,r\} \to \pm,
\] 
or equivalently the set of $r$ points 
\[
p_{i,\sigma(i)}, \  (i=1,\dots, r).
\]

Our study of Ramond punctures leads to the following:

\begin{conjecture} \label{conj} Let $V$ be a vector space  of dimension $\geq r$  containing $r$ pairs of points $p_{ij}, i=1,\dots ,r,  \  j \in \pm
$ that span $V$.  The following conditions (*), \  (**) \   are equivalent:

\noindent (*) \  \  For each of the $2^r$ sections $\sigma$, the $r$ points 
$p_{i,\sigma(i)}, \  (i=1,\dots, r)$ 
are linearly dependent.

\noindent (**) \  \  There is a subset 
$I \subset \{1, \dots , r\}$
of some cardinality $k<r$ such that the $2k$ points
$ p_{ij}, \ i \in I, \ j \in \pm$
span a subspace of dimension $<k$.
\end{conjecture}

In section \ref{pf}  of the appendix we prove

\begin{theorem} \label{LinAlg}
The conjecture holds when $V$ is $r$ dimensional.
\end{theorem}

We will need only this case. The full Conjecture has been proved in \cite{Nathan}.

 \section{From Ramond punctures to linear algebra} \label{translation}
In this section we reduce Theorem \ref{Main} to the linear algebra result, Theorem \ref{LinAlg}.

We begin by proving that condition (*) in Conjecture \ref{conj} and Theorem \ref{LinAlg} characterizes the locus $\tilde{Y}$. Let $V$ be the vector space 
\[
V := ( H^0(L) / H^0(L-D) )^*
\]
of linear functions on $H^0(L)$ that vanish on $H^0(L-D) $. 
(Equivalently, it is the dual of the image of the odd superperiod map $P_1$.)
By Riemann-Roch, its dimension is $r$.
We are interested mostly in the case that $h^0(L)=r$, so $H^0(L-D)=0 $ and $V=H^0(L)^*$.
For each $i \in \{ 1,\dots,r \}$ and $j \in \pm$, 
fix an orientation $\psi_{ij}$, i.e. a trivialization of the fiber $L_{|x_{ij}}$ of $L$ at $x_{ij}$, compatible with the residue trivialization of  $L^{\otimes 2} = K_C (D)$. 
Evaluation at $x_{ij}$ then becomes a vector $q_{ij} \in V$. 
We let $p_{ij} := q_{i+} +j \sqrt{-1}q_{i-}$, or explicitly:
\[
p_{i+} := q_{i+} +  \sqrt{-1} q_{i-}, \  \  \  p_{i-} := q_{i+} -  \sqrt{-1} q_{i-}.
\]
The collection of $2r$ vectors $q_{ij}$ clearly spans $V$, and therefore so do the $p_{ij}$. A sign choice $\sigma$ picks, for each $i$, a vector $p_{i \sigma(i)}$ which is one of $p_{i+}, p_{i-}$. 
The condition for being in $Y''_{\sigma}$ is that the $r$ points  $p_{i \sigma(i)}, \  \  i=1,\dots,r$ should be contained in a hyperplane in $V$. 
Condition (*) of  Theorem \ref{LinAlg} therefore characterizes $\tilde{Y}$. 

Next we describe a certain locus $W$ in moduli space which we  will show is equal to both $B$ (Lemma \ref{nost}) and $\tilde{Y}$ (Lemma \ref{charw}), and thereby we prove Theorem \ref{Main}.   
For every subset
\[
I \subset \{1,\dots,r \}, 
\]
the divisor 
\[
D_I := \Sigma_{i \in I}(x_{i+}+x_{i-})
\]
which is of degree $2k$ where $k := \#I$, determines  the locus:
\[
W_I := \{(C,D,L) | h^0(L(-D_I)) \geq r+1 - k \} \subset M.
\]
For example, the bad locus $B$ equals $W_{\emptyset}$. Now let
\[
W := \cup_I W_I.
\]

\begin{note} These loci obey a duality. Let $k:=\#I$ and let $I'$ denote the complement of $I$. We have 
\[
h^0(L(-D_I)) = \chi(L(-D_I)) +h^1(L(-D_I))  = r-2k + h^0(K \otimes L^{-1}(D-D_{I'})) = r-2k + h^0(L-D_{I'}))
\]
so $W_{I'} = W_I$.
\end{note}

\begin{lemma} \label{charw} Condition (**) of Conjecture \ref{conj} and Theorem \ref{LinAlg} characterizes $W$.   In particular, $W=\tilde{Y}$. 
\end{lemma}

\begin{proof} Let $k := \#I$.
The condition that the $2k$ points $p_{ij}, \  i \in I, \  j \in \pm$ are contained in a subspace of dimension $k-1$ implies that we are in $W_I$.
The converse fails in general, because $h^0(L)$ can be $>r$ if we are at a point of $B$.
However, if we are at a point of $\tilde{Y}$ that is not in $B=W_{\emptyset}$, then $h^0(L)=r$, $H^0(L-D)=0$ and $V=H^0(L)^*$, and then the condition that the $2k$ points $p_{ij}, \  i \in I, \  j \in \pm$ are contained in a subspace of dimension $k-1$ is equivalent to being in $W_I$. 
So Condition (**) of  Theorem \ref{LinAlg} characterizes $W$. 
\end{proof}

The proof of Theorem \ref{Main} is now reduced to the following result: 


\begin{lemma} \label{nost}
We have the inclusion and equality:
\[
B=W_{\emptyset} =  \cup_I W_I = W \subset Y.
\] 
In particular, $B=Y=\tilde{Y}$. 
\end{lemma}

\begin{proof} 

A point of $W_I$ is given by $(C,D,L) $ such that $h^0(L(-D_I)) \geq r+1 - \#I$. 
For elements of $H^0(L(-D_I))$, 
the conditions for being in some $Y''_{\sigma}$ are clearly satisfied automatically at points of $I$  since such an element vanishes along the divisor $D_I$,
so we only need to impose $r-\#I$ conditions at the remaining points in $D - D_I$,  
ending with $h^0(L_{\sigma}) \geq 1$. Thus $W = \cup_I W_I \subset Y$.

On the other hand, each $W_I$ is contained in $B$. 
Indeed, say we have a point of $W_I$ that is not in $B$. 
Then $h^0(L)=r$ and the divisor $D_I$ imposes at most $k-1$  conditions on sections of $L$, where we assume $1 \leq k:=\#I \leq r-1$.
But the equality $W_I = W_{I'}$ implies that the divisor $D_{I'}$ imposes at most $r-k-1$ conditions on sections of $L$. 
So $h^0(L-D) \geq r - (k-1) -(r-k-1) =2$, contradicting \eqref{Euler}.

\end{proof}

\section{Dimension counts}
Among general line bundles $L$ of degree $g-1+r$ on a curve $C$ of genus $g$, for $0 \leq r <g$, the locus where $h^0(L)>r$ has codimension $r+1$ by Brill-Noether theory \cite{BN,ACGH}. One might therefore naively expect that the bad locus $B$ should have codimension $r+1$ in $M$. The residue theorem implies that this expectation is wrong: it is off by at least 1. In this section we show that there are no further corrections:

\begin{theorem} \mbox{   } \label{dim}

\begin{itemize}
\item In the unpunctured case $r=0$, the moduli space $M=M^+ \cup M^-$ is reducible. The bad locus $B$ consists of all of $M^-$ plus the irreducible ``vanishing thetanull" divisor in $M^+$.
\item For $1 \leq r < g$, every component of $B$ has codimension $\leq r$ and there is a component of codimension exactly $r$.
\item For $r \geq g$ the bad locus is empty.
\end{itemize}
\end{theorem}

\begin{proof}
The case $r=0$ is well known: $M^-$ is the moduli of odd spin bundles $L$, where $h^0(L)$ is odd and in particular $h^0(L) \geq 1$.
$M^+$ is the moduli of even spin bundles $L$, for which $h^0(L)=0$ generically and $h^0(L)$ is even and $\geq 2$ along the irreducible divisor of vanishing thetanulls \cite{Atiyah, Mumford}.

For $r \geq g$, the degree of a twisted spin bundle $L$ is $\geq 2g-1$, so $h^1(L)=0$ and $h^0(L)=r$, so the bad locus is empty.

Assume $1 \leq r < g$. To see that the codimension is at most $r$ (and thus off by 1 from the Brill-Noether expectation), it is convenient to switch to the $2r$-sheeted cover 
\[
\tilde{M} := \{(C,D,L,p) \ | \ (C,D,L) \in M \ {\text{and}} \ p \in D \}
\]
and to $\tilde{B}$, the inverse image there of $B$.
Consider the map
\[
f: \tilde{M} \to {\text{Pic}}^{g-r}(\CC /M),
\]
sending 
\[
(C,D,L,p) \mapsto N:= L(p-D).
\]
Here $\CC$ is the universal curve over $M$. In the Picard, the effective locus 
\[
G^0_{g-r} := \{ N \ | \ h^0(N)>0 \}
\]
has codimension $r$. We claim that 
\[
\tilde{B}=f^{-1}(G^0_{g-r}),
\]
which implies that the codimension is at most $r$. 

Indeed in one direction, 
\[
L \in \tilde{B} \iff h^0(L)>r \iff h^0(L(-D)) >0 \implies h^0(N)>0.
\]
But conversely, a non-zero section $s \in H^0(N)$ gives a meromorphic 1-form 
\[
s^2 \in H^0(K(-(D-p))(p))
\]
with a possible pole only at $p$. By the residue theorem, it actually cannot have a pole there, so $s$ must be a non-zero section in $H^0(N-p) = H^0(L-D)$.

To conclude, we need to exhibit a component of codimension $r$. Let $A$ be a generic odd theta characteristic: 
\[
A^{\otimes 2} = K, h^0(A)=1.
\]
Let $a \in {\text{Sym}}^{g-1}(C)$ be the unique divisor of a non-zero section of $A$, and let $d \in {\text{Sym}}^{r}(C)$ be the sum of any $r$ of these points. (Recall we are assuming $r \leq g-1$.) Let $p \in d$ be one of these points. Generically such $p$ is not in $a-d$, which we will assume.
Take $D:=2d, L:=A(d)$. We claim that $f(\tilde{M})$ is transversal to $G^0_{g-r}$ at $f(C,D,L)$, so the component of $\tilde{B}$ through $(C,D,L)$ has codimension $r$.

This is seen by comparing the tangent spaces to $f(\tilde{M})$ and $G^0_{g-r}$ at 
\[
f(C,D,L)=N:= L(p-D)=A(d+p-2d)=A(-(d-p)).
\]
The tangent there to $G^0_{g-r}$ is spanned by the $g-r$ tangent lines to the Abel-Jacobi image of $C$ in its Jacobian at the points of $a-d+p$, or projectively by the corresponding $g-r$ points of the canonical image of $C$ in canonical space. 
The tangent to $f(\tilde{M})$ is likewise spanned by the $2r$ tangent lines to the Abel-Jacobi (or: canonical) image of $C$ at the points of $D=2d$. (Since our $D$ is non-reduced, this really means the $r$ osculating planes to $C$ at the points of $d$.) This follows from the relation 
\[
N^{\otimes 2} = K(-(D-p)+p),
\] 
which implies that the differential of $f$ at each of these points is $\pm \frac{1}{2}$ times the corresponding differential of Abel-Jacobi, so the spans are the same.

The transversality is therefore equivalent to saying that the divisor
\[
2d + (a-d+p) = a+d+p
\] 
is not contained in any canonical divisor. From our assumption that $h^0(A)=1$ it follows that the unique canonical divisor containing $a$ equals $2a$. This does not contain $a+d+p$ because of our assumption that $p \notin a-d$.

\end{proof}

\begin{appendices}
\label{appendix} 

 \section{Proof of the Linear Algebra Result} \label{pf}

We rephrase Theorem \ref{LinAlg} as follows:
   \begin{theorem} \label{main} Let $\{P_1, \dots, P_r \}, \ P_i \neq 0$ be a set of nonzero subspaces of the $r$ dimensional vector space $V= ( H^0(L) / H^0(L-D) )^*$ such that $\on{dim} P_i \le 2$.  For each $P_i$, fix two spanning  vectors $ \{p_i, p_i' \}$ (with one necessarily redundant if $\on{dim} P_i=1$) and define
\[ V_r= \{ \{ v_1, \dots, v_r \}  \  | \  v_i \in \{p_i, p_i' \} \}.  \]
If \textbf{no} element of $V_r$ spans $V$, then there exists a subset  $\{P_{i_1}, \dots, P_{i_k} \}$ for $k \le r$ such that 
\[  \on{dim} \on{span} \{ P_{i_1}, \dots, P_{i_k} \} < k. \] 

\end{theorem}

Let $\rho = \on{max}\{ \ \on{dim}(\spann v) \ | \  \forall v \in V_r \}$ and let $\vrho = \{v \in V_r  \ | \on{dim} \spann v = \rho \}$. 
For each $v \in V_r$ let the set $v' \in V_r$ be defined by the rule that if $v_i \in v$ then $v_i' \in v'$. 
Let
 \[ v_{-j} = \{v_1, \dots, \hat{v}_{j},  \dots, v_r \}, \ \ v_{+i'} = \{v_1, \dots, v_r,v_i'\}.  \] 
For $A \subset (v \cup v') \cap W$, the square brackets $[A]$ will denote the indices of the elements in $A$.
For example, if $A=\{v_1, v_1', v_2, v_3 \}$, then $[A]$ is the multi-set $\{1,1,2,3\}$ where the two copies of $1$ correspond to $v_1$ and  $v_1'$. 
We say that $i \in [A]$ is a \emph{double} if $v_i \in A$ and $v_i' \in A$.

\begin{definition} Let $v \in \vrho$, let $W= \spann v$, and let $W'=\spann v'$. We say that $ w \subset (v \cup v') \cap W $ is a \emph{spanning set} if it satisfies the following conditions: (1) $\spann w=W$, (2) $\on{card}w=\rho$, and (3) $[w]$ contains no doubles. 
\end{definition}

Note that a spanning set $w$ is just a certain choice of basis for $W$. We let $\on{max}(v)$ denote the set of all spanning sets in $(v \cup v') \cap W$ and let $w^c \subset v$ denote the complement of $w$.

\begin{lemma} \label{maxspancontainspi} Let $v \in \vrho$ and let $w \in \maxv$. If $i \in [w^c]$, then  $P_i \subset \spann w$.
\end{lemma}

\begin{proof} Suppose $P_i \not \subset \spann w$. Then $v_i' \not \in \spann w$, or $v_i \not \in \spann w$. WLOG, let  $v_i' \not \in \spann w$ and note that, under this assumption, we have that $\dim \spann w_{+i'}= \rho +1$.  Recall that we assumed that the maximal dimension of any element in $\vrho$ is equal to $\rho$. Since $w_{+i'} \in \vrho$, the equality $\dim \spann w_{+i} =\rho +1$ contradicts the maximality of $\rho$ and thus, $P_i \subset \spann w$. 
\end{proof} 

\begin{definition} \label{minset}Let $v \in \vrho$ and let $P_i \subset W=\spann v$. We say that a nonzero subset $\minpi{P_i} \subset v$ is a \emph{minimal set for $P_i$} if $P_i$ is contained in $\on{span} \minpi{P_i}$ but there exists no proper subset of $\minpi{P_i}$ whose span contains $P_i$.

\end{definition} 

If we fix $w \in \maxv$, then there is a unique minimal set $\minipi \subset w$ for $P_i= \spann \{v_i, v_i' \}$  consisting of the unique set of generators in $w$ for $v_i$ and $v_i'$, e.g, if  \[ v_i' = a_1 v_{i_1} + \cdots + a_{\rho} v_{i_{\rho}},  \]
is the linear combination for $v_i'$ in terms of the basis $w= \{v_{i_1}, \dots, v_{i_{\rho}} \}$, then $v_{i_1}, \dots, v_{i_{\rho}} \in \minipi$, and similarly for $v_i$. 
\begin{lemma} \label{pthreeandpd} Fix $w \in \maxv$ and let $i \in [\redw]$.  Let $\minipi \subset  W$ be minimal  for $P_i$, then for each $j \in [\minipi]$, we have that $P_j \subset W$.
Furthermore, 
if $\minpj$ is a minimal set for $P_j$ and $l \in [\minpj]$, then $P_l \subset W$. More generally, any $n$-th descendant of $P_i$ is contained in $W$. 
\end{lemma}
\begin{proof} 

 By definition, if $v_j \in \minipi$, then $P_i \not \subset \spann \minipi_{-j}$ and so either $v_i$ or $v_i'$ is \emph{not} in $\spann \minipi_{-j}$. WLOG, let $v_i' \not \in \minipi_{-j}$ and note that this assumption implies that $v_i' \not \in \spann w_{-j}$. Since $v_i' \not \in \spann w_{-j}$, 
the set $\m$ is a spanning set and since $j \in [\m^c]$ we find that $P_j \subset \spann \m \subset W$ by Lemma \ref{maxspancontainspi}. For the same $w \in \maxv$ as above, let $\minpj \subset w$ be the minimal set for $P_j$ and let $l \in [\minpj]$.  
We can use the same proof to prove that $P_l \subset W$ if we can find a minimal set $B$ for $P_j$ such that $l \in [B]$ and $j \not \in [B]$.  Since $P_j \subset \spann w_{-j, +i'}$, there exists a unique $B \subset \m$ minimal for $P_j$.  For $l \neq j, i'$,  note that $l \in [\minpj]$ if and only if $l \in [B]$. since we already know that $P_j$ and $P_i$ are contained in $W$, we may assume that $l \neq j,i'$. Now since $j \not \in [B]$, we can use the same proof as above to conclude that $P_l \subset W$. 

\end{proof}

Fix $w \in \maxv$ and define
\[ [\Upsilon_n] = \{ \ j \ | \ v_j \in [\minipi], \ \text{such that} \ \minipi \subset w \ \text{and} \ i \in \Upsilon_{n-1} \} \cup \Upsilon_{n-1},  \] 
where 
\[ [\Upsilon_0] = \{ \ j \ | \ v_j \in [\minipi], \ \text{such that} \  \minipi \subset w  \ \text{and} \ i \in [\redw], \} \] is the ``initial condition." Note that we meant what we wrote when we defined $[\Upsilon_0]$, i.e., we do \emph{not} want $\Upsilon_0= [\redw]$.  Clearly, the sets $\Upsilon_1, \Upsilon_2, \dots$ will eventually converge to some subset $\Upsilon$ of $w$. That is, there exists some $n_0 >> 0$ such that $ \Upsilon_{n_0}= \Upsilon_{n_0 + i_0}, \forall i_0 \ge 0$ and so $\Upsilon=\Upsilon_{n_0}$. 
 
 \begin{lemma} \label{special} If $j \in [\Upsilon]$, then $P_j \subset \spann \Upsilon$. 

\end{lemma}

\begin{proof}
  If $j \in [\Upsilon]$, then $j \in [\Upsilon_n]$ for some $n$. By construction,  $\minpj \subset \Upsilon_{n+1}$ and so $P_j \subset \spann \Upsilon_{n+1}$. In particular, $P_j \subset \spann \Upsilon$ since $\Upsilon_{n+1} \subset \Upsilon$.

\end{proof}

The proof of Theorem \ref{main} is now trivial: 

\begin{proof}[Proof of Theorem \ref{main}]  
Let
\[[\Upsilon] =\{j_1, \dots, j_h \}  \] 
and note that $\on{dim} \spann \Upsilon \le h$ and that $\spann \Upsilon$ contains at least $h$ planes, i.e, $P_{j_1}, \dots, P_{j_h} \subset \spann \Upsilon$ by Lemma \ref{special}. 
Recall that by construction $j_1, \dots, j_h \in [w]$, but $P_i \subset \spann \Upsilon$ for all $i \in [\redw]$ since we defined $\bigcup_{i \in \redw } \minipi= \Upsilon_0$. Since $\on{card} \redw \ge 1$, the set $\{P_{j_1}, \dots, P_{j_h} \} \cup \bigcup_{i \in \redw} P_i$ satisfies the conclusion of Theorem \ref{main}. 
\end{proof}


\section{Odd periods}

In \cite{WittenRamond}, the odd periods are defined to be restrictions of elements of $\Omega$ to the components of the Ramond divisor, see Section \ref{oddwitten}. Throughout the paper, we computed the odd periods  as values of elements in $H^0(C,L)$ along the points of $D$. In this section of the appendix we prove (Lemma \ref{sameperiod}) that these values agree with the values of the odd periods as defined by Witten . Throughout this section we consider only {\em split} super Riemann surfaces, e.g. ones defined over a point or more generally over a (bosonic) base.

\subsection{ $\Omega_1=H^0(C,L)$}  
We begin by recalling Witten’s argument that the odd elements of $\Omega$ can be identified with global sections of $L$. 
Let $\on{Ber}_X'$ denote the subsheaf of $\Ber(R)$ whose sections have a $\theta$-independent residue along $R$.  One can check that 
\begin{equation} \label{splits}   \Ber'= \Omega_C^1 \oplus \Pi L.  \end{equation} 
In contrast to an ordinary curve,  the dualizing sheaf $\Ber$ on the supercurve $X$ is not equal to the cotangent bundle $\Omega_X^1$: this is of course obvious given that $\Ber$ has rank $0|1$ while $\Omega_X^1$ has rank $1|1$. However, 
there is a canonical isomorphism $\alpha: H^0(X, \Ber') \cong \Omega$ which in local superconformal coordinates $(z, \theta)$ sends $s=(g(z) + \theta f(z) ) [dz | d \theta]$ to 
\begin{equation}  \label{standard} \alpha(s) = g(z) d \theta + (f(z)  + \theta g'(z)) dz. \end{equation} 
Composing $\alpha^{-1}$ with \eqref{splits} we find that $\Omega= H^0(C, \Omega_C^1) \oplus \Pi H^0(C,L)$ (we will write $=$ whenever an isomorphism is canonical), and so  $\on{dim} \Omega = \on{dim} H^0(X, \Ber')=g| h^0(L)$.


\subsection{Odd periods} \label{oddwitten} One can check that 
around each $R_i \in R$, there exist superconformal coordinates $(z, \theta)$ such that $R_i$ is defined by the equation $z=0$.  In these coordinates a general element of $\Omega$ is of the form $f(z, \theta) dz + g(z, \theta) d \theta$ and its restriction to $R_i$ (to compute: set $z=0$ and apply the condition $ds=0$)  can be shown to be equal to 
\[  g(0) \ d \theta, \]  
where the constant $g(0) \in \C$ is unique up to a superconformal change of coordinates $\theta \mapsto  - \theta$.  Let $\sigma= \{ \sigma_i \in \{ +, - \} \}$ denote a choice of $\theta$ or $-\theta$ for each $R_i \in R$. 

\begin{definition} For each $R_i \in R$, the \emph{odd period $\frak{o}_i^{\sigma}(s)$ of $s \in \Omega$ with respect $\sigma$ } is the constant $\sigma_i \cdot g(0)$.
\end{definition}

\begin{lemma}  \label{vanishing} The odd (resp. even) periods of even (resp. odd) elements in $\Omega$ are equal to zero.  
\end{lemma} 

\begin{proof}  If $s$ is an even element in $\Omega$, then its restriction to $R_i$ must also be even; however $g_0 d \theta$ is odd for all non-zero $g_0$. Thus $g_0=0$. 

Now let $s \in \Omega_1$ be odd and let $\{a_1, \dots, a_g, b_1, \dots, b_g \}$ be A- and B-cycles on $C$.  The even (A- and B-) periods of $s$ are the integrals 
\[  \int_{a_i} \tau^*(s), \ \ \int_{b_i} \tau^*(s) \] 
where $\tau: C \subset X$ is the canonical inclusion. In local superconformal coordinates $(z, \theta)$, the odd closed one form $s$ can be written as $s= f(z) \theta dz + g(z) d \theta$. The pullback $\tau^*(s)$ is equal to $0$ since $\tau^*(\theta)=0$ and thus both of the above integrals are equal to zero. 
\end{proof} 

We will now show that odd periods can be computed as residues along $D$. We need to explain what we mean by a residue of an element in $\Omega$. 
Using the identification $\Omega= H^0(X, \Ber')$ and the splitting \eqref{splits} we can write every element $s$ in $\Omega$ as a sum $s_0+ s_1$ of a holomorphic differential $s_0$ on $C$ and a global section $s_1$ of the twisted spin structure $L$.  For each $p_i \in D$, we define \[ \on{res}_i(s)= \on{res}_{p_i}(s_0) + \on{res}_{p_i}(s_1)=\on{res}_{p_i}(s_1) \] 
where $\on{res}_{p_i}(s_0)=0$ because $s_0$ is holomorphic near $p_i$. Since $s_1$ is not a differential, we need to explain what we mean by $\on{res}_{p_i}(s_1)$: Locally near each $p_i \in D$, there exists coordinates $U(z)$ such that $\Omega_U^1 \cong \mc{O}_U$ and such that $p_i$ is defined by $z=0$. Let $dz$ be a generator for $\Omega_U^1$.  We then have two possible generators for $L \vert_U$, namely $\theta:=\sqrt{dz/z}$ and $-\theta:= - \sqrt{dz/z}$. For a fixed generator $\theta$, there is a canonical isomorphism $L \vert_U = \Omega_U^1(D)=\Omega_U^1(p_i)$ sending $\theta$ to $dz/z$.  
Let $\sigma = \{ \sigma_i \in \{ +, - \} \}_{i=1}^{2r} $ denote the sign of the generator for  each $L \vert_{U_i}$. Then each $s \in \Omega_1$ is locally of the the form $\sigma_i \cdot g(z) \theta$ and is sent to $\sigma_i \cdot \frac{g(z)}{z} dz$ in $\Omega_U^1(p_i)$: And, 
\begin{equation} \label{omegaind} \on{res}_{p_i}(s_1) =\sigma_i \cdot  \on{res}_{p_i}
\left (\frac{g(z)}{z} dz \right).  \end{equation}



\begin{lemma} \label{sameperiod} Let $\sigma$ be fixed and let $s \in \Omega$. Then for each point $p_i \in D$, we have \[ (1/ 2 \pi i) \frak{o}_{p_i}(s)= \on{res}_{p_i}(s). \]  
\end{lemma} 

\begin{proof} The lemma is obvious for even elements $\Omega$ so let $s$ be an odd element and let $U(z, \sigma \cdot \theta)$ be superconformal coordinates near $R_i$ so that $R_i$ is defined by $z=0$. Then there exist unique $g$ such that 
\[ s \vert_U = \sigma_i \cdot ( g(z) d \theta + g'(z) \theta dz) \] 
and $\alpha^{-1}(s)= ( \sigma_i \cdot g(z) ) [dz | d\theta] $.  The odd periods are now easy to read off: namely, $\frak{o}_i^{\sigma}(s)=\sigma_i \cdot g(0)$. 

The section
$\alpha^{-1}(s)$ is identified with the section $\sigma_i \cdot g(z) \theta$ of $L \vert_U$ by \eqref{splits} where $\theta= \sqrt{\frac{dz}{z}}$. The section  $\sigma_i \cdot g(z) \theta$  is in turn identified with the section $\sigma_i \cdot \frac{g(z)}{z}  dz$ of $\Omega_U^1(p_i)$ under the identification $L \vert_U = \Omega_U^1(p_i)$. It is now easy to see that \[ 2 \pi i \on{res}_{z=0} \left (\frac{g(z)}{z} dz \right)= \sigma_i \cdot g(0)=\frak{o}_i^{\sigma}(s) \] .

\end{proof}

\end{appendices}

\end{document}